\documentclass[a4paper]{article}

\usepackage[english]{babel}
\usepackage{hyperref}

\usepackage{amsmath, amsfonts,xfrac, amsthm,amssymb,fancyhdr}

\usepackage{mathtools}

\usepackage{setspace}

\usepackage{geometry}
\usepackage{xypic,colortbl}

\usepackage{fixltx2e, ifthen, sidecap, wrapfig, xspace}

\usepackage{mathrsfs, stmaryrd}

\usepackage{sectsty}

\usepackage{
 calc, enumitem, makeidx,bbold
}%showidx, bbold

\usepackage[font=small,format=plain,labelfont=sc,textfont=it]{caption}
\usepackage[T1]{fontenc}
\usepackage[utf8]{inputenc}
\usepackage[all]{xy}
\usepackage[geometry]{ifsym}
\usepackage[cmyk]{xcolor}

\usepackage{comment}

\newcommand{\mathsym}[1]{{}}

{\begin{figure}[#1]
\centering 
\includegraphics[scale=#3]{#2}
}
{\end{figure}}

\newlist{enumeratep}{enumerate}{10}
\setlist[enumeratep]{label=\quad\textit{\arabic*'.},ref=\arabic*',leftmargin=*}

\newenvironment{romanlist'}[0]
{\begin{list}{\makebox[0.5cm][l]{\textit{\roman{enumi}')}}}{\usecounter{enumi}}}
{\end{list}}

\newcommand{\savelabel}[2]{\expandafter\newtoks\csname#1\endcsname
  \global\csname#1\endcsname={#2} \label{#1} #2}
\newcommand{\loadlabel}[1]{\noindent {\bf Lemma~\ref{#1}. } \textit{\the\csname#1\endcsname}
\medskip

}
\renewcommand{\setminus}{-}%{\rotatebox[origin=c]{-20}{--- }}

\newcommand{\loadlabelthm}[1]{\medskip\noindent {\bf Theorem~\ref{#1}. }
  \noindent  \textit{\the\csname#1\endcsname}
\medskip
}

\newcommand{\loadlabelprop}[1]{\noindent {\bf Proposition~\ref{#1}. }
  \noindent  \textit{\the\csname#1\endcsname}
\medskip

}

\newcommand{\eps}{\varepsilon}

\renewcommand{\subset}{\subseteq}

\newcommand{\atleast}[1]{{\ge n}}
\newcommand{\less}[1]{{<n}}

	\newcommand{\notacol}[2]{}
\newsavebox{\quoteitbox}

{\hspace*{\fill}{\upshape(\usebox{\quoteitbox})}\end{quote}%
\end{sloppypar}\noindent\ignorespacesafterend}

\newenvironment{quoteit*} 
{\begin{sloppypar}\noindent\slshape\begin{quote}\itshape} 
	{\end{quote}\ignorespaces\end{sloppypar}\noindent\ignorespacesafterend}

\newenvironment{quotetag*}
{~\par%\vskip 0mm                % Skip down 1/2 before story
	\begingroup                  % Start of formatting properties
	\begin{equation*}
		 \begin{minipage}[c]{115mm}
			\it\noindent{\par}
}
{
		\end{minipage}
	\end{equation*}
	\endgroup                        % End of formatting properties
\par
\textnormal
\medskip
}

\newcommand\set[1]{\ensuremath{\{#1\}}}

\newtheoremstyle{theoremstyle}% name of the style to be used
  {3pt}% measure of space to leave above the theorem. E.g.: 3pt
  {3pt}% measure of space to leave below the theorem. E.g.: 3pt
  {\itshape}% name of font to use in the body of the theorem
  {0pt}% measure of space to indent
  {\bfseries}% name of head font
  {}% punctuation between head and body
  {4pt}% space after theorem head; " " = normal interword space
  {}% Manually specify head

\theoremstyle{theoremstyle}
\newtheorem{theorem}{Theorem}[section]
\newtheorem*{theorem*}{Theorem}

\newtheorem{lemma}[theorem]{Lemma}%[chapter]
\newtheorem{corollary}[theorem]{Corollary}%[chapter]
\newtheorem{proposition}[theorem]{Proposition}%[chapter]
\newtheorem{claim}{Claim}

\newtheoremstyle{remarkstyle}% name of the style to be used
  {3pt}% measure of space to leave above the theorem. E.g.: 3pt
  {10pt}% measure of space to leave below the theorem. E.g.: 3pt
  {}% name of font to use in the body of the theorem
  {0pt}% measure of space to indent
  {\itshape}% name of head font
  {}% punctuation between head and body
  {4pt}% space after theorem head; " " = normal interword space
  {\thmname{#1}\thmnumber{ #2}\thmnote{ (#3)}.}% Manually specify head

\theoremstyle{remarkstyle}
\newtheorem{example}{Example}[section]

\newtheoremstyle{definitionstyle}% name of the style to be used
  {3pt}% measure of space to leave above the theorem. E.g.: 3pt
  {3pt}% measure of space to leave below the theorem. E.g.: 3pt
  {}% name of font to use in the body of the theorem
  {0pt}% measure of space to indent
  {\itshape}% name of head font
  {}% punctuation between head and body
  {4pt}% space after theorem head; " " = normal interword space
  {\thmname{#1}\thmnumber{ #2}\thmnote{ (#3)}.}% Manually specify head

\theoremstyle{definitionstyle}

\numberwithin{equation}{section}

\newlength{\wideaslength}

\renewcommand{\subset}{\subseteq}

\renewcommand{\models}{\vDash}
\newcommand{\nmodels}{\nvDash}

\newcommand{\seta}[1]{}

\def\lsim{\mathrel{\rlap{\lower4pt\hbox{\hskip1pt$\sim$}}
    \raise1pt\hbox{$<$}}}                % less than or approx. symbol

\definecolor{gray1}{rgb}{0.99,0.99,0.99}
\definecolor{gray2}{rgb}{0.97,0.97,0.97}
\definecolor{gray3}{rgb}{0.95,0.95,0.95}
\definecolor{gray4}{rgb}{0.93,0.93,0.93}
\definecolor{gray5}{rgb}{0.91,0.91,0.91}
\definecolor{gray6}{rgb}{0.89,0.89,0.89}
\definecolor{gray7}{rgb}{0.87,0.87,0.87}
\definecolor{gray8}{rgb}{0.85,0.85,0.85}
\definecolor{gray9}{rgb}{0.83,0.83,0.83}
\definecolor{gray10}{rgb}{0.81,0.81,0.81}
\definecolor{gray20}{rgb}{0.71,0.71,0.71}
\definecolor{gray40}{rgb}{0.51,0.51,0.51}

\usepackage{multicol}

\newcommand{\implication}[2]{(\ref{#1}$\rightarrow$\ref{#2})}

\newcommand{\cont}{\mathfrak c}

\newcommand{\ultra}{\mathcal{U}}
\newcommand{\meas}{\mu}
\newcommand{\uae}{$\meas$-almost-everywhere\xspace}
\newcommand{\Nat}{\mathbb N}
\newcommand{\seq}[1]{\mathbf{#1}}
\newcommand{\ultraprod}{\prod_\ultra}
\newcommand{\cl}{\mathcal C}
\newcommand{\cls}{\mathcal C^\star}
\newcommand{\gr}{\mathcal G\mathit{raphs}}
\newcommand{\nab}{\mathop{\triangledown}}
\newcommand{\topnab}{\mathop{\tilde \triangledown}}

\def\cqedsymbol{\ifmmode$\lrcorner$\else{\unskip\nobreak\hfil
\penalty50\hskip1em\null\nobreak\hfil$\lrcorner$
\parfillskip=0pt\finalhyphendemerits=0\endgraf}\fi} 

\newcommand{\cqed}{\renewcommand{\qed}{\cqedsymbol}}

\title{On ultralimits of sparse graph classes%
\thanks{During the work on this project, Micha\l{} Pilipczuk has been a post-doc at Warsaw Centre of Mathematics and Computer Science, and has been supported by the Foundation for Polish Science via the START stipend programme.}}
\author{
  Micha\l{} Pilipczuk
  \thanks{
    Institute of Informatics, University of Warsaw, Poland, \texttt{michal.pilipczuk@mimuw.edu.pl}.
  }
  \and
  Szymon Toru\'{n}czyk\thanks{
    Institute of Informatics, University of Warsaw, Poland, \texttt{szymon.torunczyk@mimuw.edu.pl}.
  }
  }

\date{}

\begin{document}

\begin{titlepage}
\def\thepage{}
\thispagestyle{empty}
\maketitle

\begin{abstract}
The notion of nowhere denseness is one of the central concepts of the recently developed theory of sparse graphs. We study the properties of nowhere dense graph classes by investigating appropriate limit objects defined using the ultraproduct construction. It appears that different equivalent definitions of nowhere denseness, for example via quasi-wideness or the splitter game, correspond to natural notions for the limit objects that are conceptually simpler and allow for less technically involved reasonings.

\end{abstract}
\end{titlepage}

\newpage
\section{Introduction}\label{sec:intro}
The theory of sparse graphs concentrates on defining and investigating combinatorial measures of sparsity for graphs, as well as more complicated relational structures. Many such abstract notions of sparsity turn out to be natural concepts connected to fundamental problems in graph theory, logics, and algorithmics. We refer to the book of Ne\v{s}et\v{r}il and Ossona de Mendez~\cite{sparsity} for a broad introduction to the field. 

In this theory, the central role is played by the notion of {\em{nowhere denseness}}. Intuitively, a class of finite graphs $\cl$ is {\em{nowhere dense}} if for any graph $G\in \cl$, after performing any ``local'' contractions in $G$ one cannot obtain an arbitrarily large clique (see Section~\ref{sec:nowhere} for a formal definition). The fundamentality of this concept is confirmed by the fact that there are multiple properties of $\cl$, reflecting different combinatorial aspects of sparsity, that are equivalent to $\cl$ being nowhere dense. For example, nowhere denseness can be equivalently defined using the asymptotics of the behavior of edge density in graphs from $\cl$ under ``local'' contractions; using various graph parameters connected to coloring and ordering; or by studying how large sets of pairwise distant vertices can be found in the graphs from $\cl$. For this work the most important is the notion of {\em{quasi-wideness}}. Intuitively, $\cl$ is quasi-wide if in any large enough graph from $\cl$ one can delete a small number of vertices so that the remaining graph contains many vertices that are far apart from each other. It is known~\cite{sparsity} that quasi-wideness is equivalent to nowhere denseness for hereditary graph classes $\cl$.

The reader might have observed that in the previous paragraph we were purposely very vague when discussing nowhere denseness and quasi-wideness. There is a good reason for this: the formal definitions of these notions require a precise formalization of intuitive concepts like ``local contractions'', ``arbitrarily large clique'', or ``large enough graph''. This applies to many notions from the theory of sparse graphs; it is usual that a formal definition of a notion involves introducing several parameters that are related to each other via a confusing sequence of quantifications. For instance, the formal definition of quasi-wideness, which can be found in Section~\ref{sec:quasi-def}, involves $4$ parameters and $6$ alternating quantifiers. For this reason, the theory of sparse graphs is known for its complicated technical layer, which often obfuscates otherwise natural reasonings.

In this work we investigate the notion of nowhere denseness using the ultraproduct construction. Our main motivation is the general meta-approach proposed by Tao in a post on his blog~\cite{terryblog}, which can be summarized as follows: In order to study a class of discrete objects, define a corresponding class of limit objects that inherits the properties of the original class. Then a reasoning for the discrete class can be translated to a reasoning for the limit class, which often turns out to be more natural and avoid a number of technicalities that were relevant in the discrete case. For discrete structures where logics play an important role, Tao proposes the usage of ultraproducts as the method of constructing limit objects. The properties of discrete and limit classes can be then translated to each other by the means of Łoś's theorem.

\paragraph*{Our contribution.} Given a class $\cl$ of finite graphs, we define the {\em{class limit}} $\cls$ as follows: $\cls$ comprises all the subgraphs of ultraproducts of sequences of graphs from $\cl$. Thus, $\cls$ contains $\cl$ as a subclass, but it also contains infinite graphs that somehow model the behaviour of increasing sequences of graphs from $\cl$. Given this definition, it is natural to ask what properties of $\cls$ correspond to the assumption that $\cl$ is nowhere dense or quasi-wide. More precisely, we investigate three definitions of nowhere denseness that are known to be equivalent: (1) the classic definition, both using shallow minors and shallow topological minors, (2) quasi-wideness, (3) the definition via the splitter game due to Grohe et al.~\cite{GroheKS14}. We show that these notions have natural limit variants that are conceptually simpler, and (with some technical caveats) can be proved to be equivalent to the standard ones using Łoś's theorem. For instance, limit quasi-wideness can be defined as follows: for every natural number $d$ and every infinite graph $G\in \cls$, one can delete a finite number of vertices from $G$ so that the remaining graph contains an infinite set of vertices that are pairwise at distance at least $d$ from each other. Note that now this definition is fully formal, and vague intuitions of small/large from the discrete case are replaced by simply finite/infinite. 

Then, we investigate the known proofs of the equivalence between the studied notions. We show that these proofs can be conveniently translated to the limit case, and again conceptual simplification occurs. For the most difficult of the studied implications, from nowhere denseness to uniform quasi-wideness, the proof for the discrete variant involves multiple usage of Ramsey's theorem, which results in the need of tracking a number of mutually related parameters throughout the reasoning. In the limit setting, however, we can use the infinite version of Ramsey's theorem. Instead of setting precise bounds on the sizes of considered objects, we can just reason whether they are finite or infinite. In some sense, this corresponds to formalizing the intuition of the proof of the discrete case by moving to the limit setting and replacing small/large with finite/infinite. 

Figure~\ref{fig:diagram} depicts all the implications proved in this paper. Note that the net of implications shows that all the depicted notions are equivalent. Observe also that, apart from the implication from quasi-wideness to nowhere denseness, we do not use any non-trivial implication between the notions in the discrete setting. Thus, we give an alternative proof of the equivalence of the studied notions of sparsity: we translate the discrete notions to the limit variants using Łoś's theorem, and then prove the equivalence in the cleaner limit setting. The fact that we use the known implication from quasi-wideness to nowhere denseness to complete our picture is due to technical complications when linking standard and limit variants of quasi-wideness using Łoś's theorem.

\paragraph*{Outline.} In Section~\ref{sec:prelims} we set up the notation, recall basic definitions and facts about ultraproducts and sparse graphs, and define key notions. In Section~\ref{sec:nowhere} we introduce the limit variants of nowhere denseness, both using shallow minors and shallow topological minors. We prove that these notions are equivalent to the standard ones and to each other. In particular, this gives an alternative proof of the equivalence of the standard definitions of nowhere denseness using shallow minors and shallow topological minors. In Section~\ref{sec:quasi} we introduce the limit variants of quasi-wideness, and prove that they are equivalent to the standard definition and to limit nowhere denseness. Section~\ref{sec:game} is devoted to the study of the limit variant of the splitter game of Grohe et al.~\cite{GroheKS14}, again with all the appropriate equivalences proved. In Section~\ref{sec:conc} we gather all the findings in one theorem and conclude with some discusion and prospects of future work.

\begin{figure}[t]
                \centering
                \def\svgwidth{\columnwidth}
                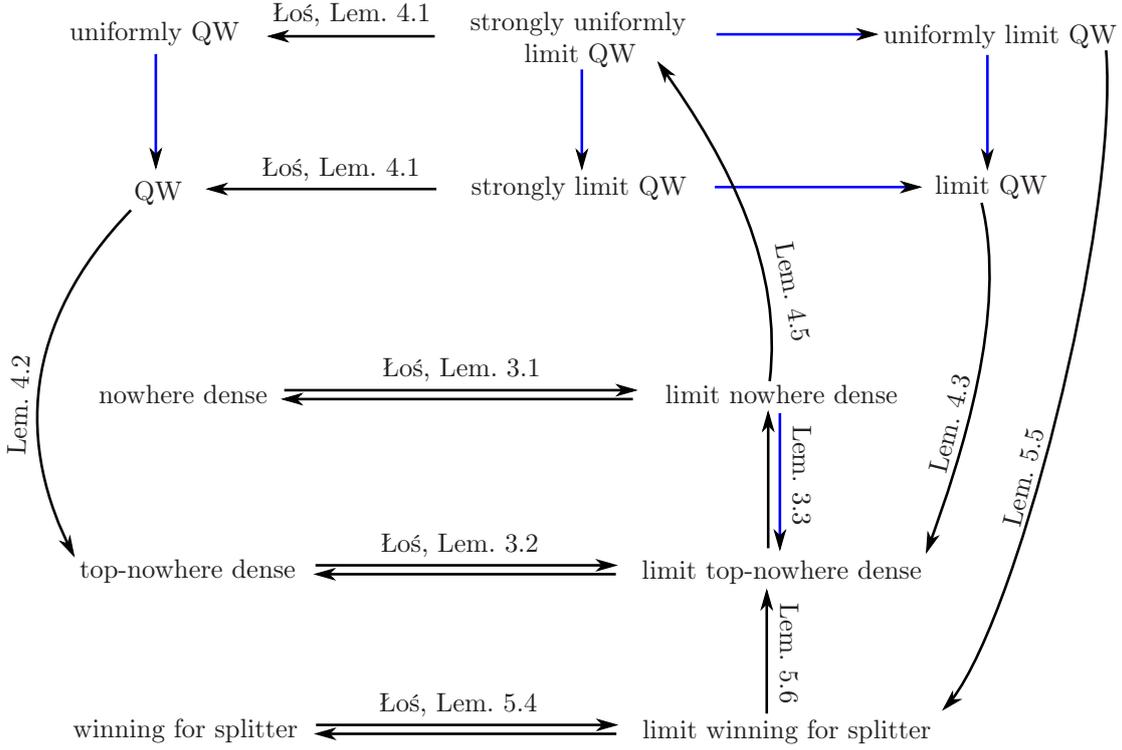
\caption{The net of implications proved in this paper. Blue arrows stand for trivial implications, QW stands for quasi-wide. The implication of Lemma~\ref{lem:finite-qw} is the only one that requires the graph class $\cl$ to be hereditary.}\label{fig:diagram}
\end{figure}

\section{Preliminaries}\label{sec:prelims}
\paragraph*{Notation.} For a set $X$, by $2^X$ we denote the family of all the subsets of $X$. By $\Nat=\{0,1,2,\ldots\}$ we denote the set of natural numbers.
 As usual, $\aleph_0$ is the cardinality of $\Nat$ and $\cont$ is the cardinality of $2^{\Nat}$. For notational convenience, we sometimes use the ordinal $\omega$ instead of the cardinal $\aleph_0$ (e.g. in subscripts), for example, $K_\omega$ denotes the countably infinite clique.
 
 Whenever $\sim$ is some equivalence relation, then $[e]_\sim$ denotes the equivalence class of $e$ w.r.t. $\sim$. We usually use boldface small latin letters to denote sequences indexed by natural numbers. We follow the convention that if $\seq g$ is such a sequence, then $\seq g = (\seq g_n)_{n\in \Nat}$, i.e., $\seq g_n$ is the $n$-th term of $\seq g$.

The graphs considered in this paper are simple, undirected and possibly infinite, unless explicitly stated. For a graph $G$, by $V(G)$ and $E(G)$ we denote its vertex and edge set, respectively. By $N_G(u)=\set{v\colon uv\in E(G)}$ and $N_G[u]=N_G(u)\cup \{u\}$ we denote the open and closed neighborhoods of $u$, respectively. For $d\in \Nat$, $N^d_G[u]$ denotes the {\em{ball od radius $d$ around $u$}}, i.e., the set of vertices of $G$ that are at distance at most $d$ from $u$. We omit the subscript whenever the graph is clear from the context.

\paragraph*{Ultrafilters.}
A family $\ultra\subseteq 2^{\Nat}$ is called an {\em{ultrafilter}} if the following conditions hold:
\begin{enumerate}
%\item $\ultra\neq \emptyset$;
\item If $X\in \ultra$ and $X\subseteq Y$, then also $Y\in \ultra$;
\item If $X,Y\in \ultra$, then also $X\cap Y\in \ultra$;
\item For every $X\subseteq \Nat$, exactly one of the sets $X$ and $\Nat\setminus X$ belongs to $\ultra$.
\end{enumerate}
Given an ultrafilter $\ultra$, we may define the corresponding measure $\meas_\ultra\colon 2^\Nat\to \set{0,1}$ by setting $\meas_\ultra(X)$ to be equal to $1$ if $X\in \ultra$, and to $0$ otherwise. From the properties of an ultrafilter it easily follows that $\meas_\ultra$ is a finitely-additive $\set{0,1}$-measure on $\Nat$.

We say that an ultrafilter $\ultra$ is {\em{non-principal}} if $\{x\}\notin \ultra$ for any $x\in \Nat$; equivalently, $\ultra$ does not contain any finite set. The existence of a non-principal ultrafilter on $\Nat$ easily follows from Zorn's lemma. For the entire paper, let us fix some non-principal ultrafilter $\ultra\subseteq 2^{\Nat}$ and the corresponding measure $\meas=\meas_\ultra$. If $\phi(n)$ is a first order formula with a free variable $n$ ranging over $\Nat$, then we say that $\phi$ holds \emph{\uae} if $\meas(\set{n:\phi(n)})=1$.

%For the entire paper, fix an \emph{ultrafilter} $\ultra$ on $\Nat$,
%i.e., a family of subsets $\ultra\subset P(\Nat)$
%such that the characteristic function $1_\ultra:P(\Nat)\to \set{0,1}$
%is monotone (with respect to $\subset$ and $\le$), finitely additive
%(with respect to disjoint union and $+$) and non-constant. 
%We say that a set $X\subseteq \Nat$ has $\ultra$-measure $0$ 
%if $1_\ultra(X)=0$, and has $\ultra$-measure  $1$ if $1_\ultra(X)=1$.

\paragraph*{Relational structures.}
\newcommand{\str}[1]{\mathbb #1}
Fix a relational language $\Sigma$, consisting of \emph{relation symbols}, each with an assigned arity – a positive natural number. A \emph{relational structure} $\str A$ over the language $\Sigma$ consists of a universe $A$, together with an interpretation mapping which assigns to each symbol $R\in \Sigma$ its \emph{interpretation} in $\str A$, which is a relation $R_\str A\subset  A^r$, where $r$ is the arity of $R$.

A graph $G$ can be treated as a relational structure $\str G$ over the language consisting of one relation symbol $E$ of arity $2$, whose universe is  $V(G)$, and where $E_\str G=E(G)$. In the following, we will treat graphs as relational structures as described above, whenever convenient.

\paragraph*{Ultraproducts.}
For a sequence of sets $(A_n)_{n\in\Nat}$, consider the equivalence relation $\sim_\ultra$ on the product $\prod_{n\in \Nat} A_n$ defined as follows: if $\seq a,\seq b\in \prod_{n\in \Nat} A_n$, then $\seq a\sim_\ultra \seq b$ iff $\seq a_n=\seq b_n$ holds \uae. Then we define the {\em{ultraproduct}} $\prod_{\ultra} A_n$ as the quotient set $\prod_{n\in \Nat} A_n/\sim_\ultra$.

Suppose now that we are given a sequence of relational structures $(\str A_n)_{n\in\Nat}$ over the same relational language $\Sigma$. We now define a new relational structure over $\Sigma$, called the {\em{ultraproduct}} of $(\str A_n)_{n\in\Nat}$, and denoted $\ultraprod \str A_n$.
Its universe is equal to $\prod_{\ultra} A_n$. 
The interpretation $R_{\ultraprod \str A_n}$ of a symbol $R\in\Sigma$ of arity $r$ is  defined as follows.
Let $\seq v^1,\ldots,\seq v^r$ be representatives of some classes of abstraction w.r.t. $\sim_{\ultra}$, then we put
$$([\seq v^1]_{\sim_\ultra},\ldots,[\seq v^r]_{\sim_\ultra})\in R_{\ultraprod \str A_n}\quad\iff \quad  (\seq v^1_n,\ldots, \seq v^r_n)\in R_{\str A_n}\textrm{ holds \uae}.$$
Note that this definition does not depend on the choice of representatives $\seq v^1,\ldots, \seq v^r$.

Most of the time, we will apply the ultraproduct construction to a sequence $(G_n)_{n\in \Nat}$ of graphs, treated as relational structures. The resulting structure is a relational structure over the 
relational language consisting of one binary relational symbol $E$, and it is not hard to see (and follows from Theorem~\ref{thm:los} below) that it is in fact a simple graph, i.e., the interpretation of the relation $E$ is symmetric and anti-reflexive. We denote this graph by $\ultraprod G_n$.

%For a sequence of graphs $(G_n)_{n\in\Nat}$, consider 
%the graph $\prod_{n\in\Nat}G_n$, whose
%vertices are $\Nat$-indexed sequences $\seq g$ such that $g_n\in V(G_n)$ for $n\in \Nat$, and edges are of the form $(\seq g,\seq h)$
%such that $E(g_n,h_n)$ holds for all $n\in\Nat$.
% Define an equivalence relation $\sim_\ultra$ on vertices so that
% $\seq g\sim_\ultra \seq h$ if and only if $g_n=h_n$ holds \uae. 
% The \emph{ultraproduct} of the sequence $(G_n)_{n\in\Nat}$
% is the graph $$\ultraprod G_n= \left(\prod_{n\in\Nat}G_n\right)/\sim_\ultra.$$
%In other words, its vertices are $\sim_U$-equivalence classes of 
%sequences $\seq g$, and there is an edge from the equivalence class
%of $\seq g$ to the equivalence class of $\seq h$
%if and only if $E(g_n,h_n)$ holds \uae.

\begin{example}
  Fix a finite graph $G$.
  If $G_n=G$ for all $n\in \Nat$, then $\ultraprod G_n$ is isomorphic to $G$.
  Indeed, every sequence $\seq g\in \prod_{n\in \Nat} V(G_n)$ is $\sim_\ultra$-equivalent to some constant sequence
  of vertices of $G$, because by the properties of an ultrafilter there is exactly one vertex $u\in V(G)$ such that $\meas(\set{n\colon \seq g_n=u})=1$.
\end{example}

\begin{example}\label{ex:cliques}
  Let $(K_{n+1})_{n\in \Nat}$ be the sequence of cliques of increasing sizes.
  Then we claim that $\ultraprod K_{n+1}$ is a clique of cardinality $\cont$. It is immediate that $\ultraprod K_{n+1}$ is a clique and has size at most $\cont$, because the cardinality of the entire product $\prod_{n\in \Nat} V(K_{n+1})$ is $\cont$. It remains to prove that $\sim_\ultra$ has at least $\cont$ equivalence classes on $\prod_{n\in \Nat} V(K_{n+1})$.

For every pair of integers $k,\ell$ with $2^\ell\le k <2^{\ell+1}$,
  choose an arbitrary injective function $\phi^\ell_k:\set{0,1}^\ell\to V(K_k)$,
   and extend $\phi^\ell_k$ to $\set{0,1}^\Nat$
  by ignoring all but the first $\ell$ positions of the sequence.
   For a sequence $\seq x\in \set{0,1}^\Nat$,
   construct a sequence $\seq g^{\seq x}$ as follows: for every $n\in \Nat$, select the unique integer $m$ such that $2^m\le n+1< 2^{m+1}$, and put
   $\seq g^{\seq x}_n=\phi^m_{n+1}(\seq x)$.
   Then, for every $n\in\Nat$, we have that $\seq g^{\seq x}_n=\seq g^{\seq y}_n$ if and only if $\seq x_m=\seq y_m$
   for all $m$ with $2^m\le n+1$. It follows that for $\seq x\neq \seq y$, the sequences $\seq g^{\seq x}$ and $\seq g^{\seq y}$ agree only on a finite number of positions, and hence $\seq g^{\seq x}\sim_\ultra\seq g^{\seq y}$
   iff $\seq x=\seq y$. Therefore, $\sim_\ultra$ indeed has at least  $\cont$ different equivalence classes.
\end{example}

%\smallskip
% Suppose we have a sequence of graphs $(G_n)_{n\in \Nat}$ and a sequence of subsets of vertices $(X_n)_{n\in \Nat}$, where $X_n\subseteq V(G_n)$ for each $n\in \Nat$. Then we can define $\ultraprod X_n\subseteq V(\ultraprod G_n)$ as the set of all the vertices $[\seq x]_{\sim_\ultra}\in V(\ultraprod G_n)$ for which $\seq x_n\in X_n$ holds \uae. Note that this condition is well-defined, because it does not depend on the choice of the representative $\seq x$.

We now recall the main tool for transferring results between sequences of relational structures and their ultraproducts. Intuitively, it says that the definition of interpretation of relational symbols in the ultraproduct implies a similar behavior for any fixed first-order formula. 

\begin{theorem}[Łoś's theorem\footnote{Jerzy Łoś was a Polish mathematician, whose surname is pronounced roughly as ``wosh''.}]\label{thm:los}
Let $(\str A_n)_{n\in\Nat}$ be a sequence of relational structures over the same language $\Sigma$, and let $\str A^\star=\ultraprod \str A_n$. Consider a first-order formula $\phi$ in language $\Sigma$, with free variables 
    $x^1,x^2,\ldots,x^k$. Then,
    for any $\seq x^1,\ldots,\seq x^k$ in the universe of $\str  A^\star$,
     $$\str A^\star,\seq x^1,\ldots,\seq x^k\models \phi\quad\iff \quad 
     \str A_n,\seq x^1_n,\ldots,\seq x^k_n\models \phi\text{ holds \uae}.$$
\end{theorem}

% \begin{remark}\label{rem:los}
% Theorem~\ref{thm:los} is a special case of a more general statement that holds for arbitrary, multi-sorted structures. In this more general result, $(G_n)_{n\in\Nat}$ can be a sequence of first-order structures over a multi-sorted language $\Sigma$, and the ultraproduct $G^\star=\ultraprod G_n$ is a structure over the language $\Sigma$ whose definition is analogous to the definition of the ultraproduct of graphs given above. The formula $\phi$ can be any first-order formula over in the language $\Sigma$. The conclusion of the theorem is the same as above.
% This general form is not relevant for the formulation of the main definitions and results in this paper, but will be useful in the proofs of Lemma~\ref{lem:ultrainclusion} and Lemma~\ref{lem:quasi-los}.

% However, at one point we will need the following fact: In Theorem~\ref{thm:los} one can enrich the graph signature with some subsets of vertices $X^1,X^2,\ldots,X^p$, regarded as unary relations. Then each graph $G_n$ is equipped with some $X^1_n,X^2_n,\ldots,X^p_n\subseteq V(G_n)$, and we define $X^{j,\star}=\ultraprod X^j_n$, for $j=1,2,\ldots,p$. The formula $\phi$ may use atomic subformulas of the form $z\in X^j$, where $z$ is a vertex variable. The theorem then states that $G^\star,(X^{j,\star})_{j=1,\ldots,p},(\seq x^j)_{j=1,\ldots,k} \models \phi$ if and only if $G_n,(X^{j}_n)_{j=1,\ldots,p},(x^j_n)_{j=1,\ldots,k} \models \phi$ holds \uae.
% \end{remark}

\begin{example}\label{ex:cycles}
    Consider the ultraproduct $C^\star = \ultraprod C_n$, where $C_n$ is the cycle of length $n$ (for $n<3$ put any graphs instead, as they do not influence the resulting ultraproduct).
    We will show that $C^\star$ is a disjoint union bi-infinite paths. 
    The same argument as in Example~\ref{ex:cliques} shows that the cardinality of $V(C^\star)$ is $\cont$.

    Consider the first-order sentence $\phi$ which says: ``every vertex has exactly two neighbors''. This sentence holds in every $C_n$. Therefore, by Łoś's theorem, it also holds in $\ultraprod C_n$. Now consider the first-order sentence $\phi_m$ (for $m\ge 3$) which says: ``there is no simple cycle of length $m$''. The s sentence $\phi_m$ holds in every $C_n$ with 
    $n>m$. In particular, for fixed $m$ and variable $n$, $C_n\models \phi_m$ holds \uae.
    Therefore, from Łoś's theorem we infer that $C^\star\models \phi_m$, for every $m\ge 3$.
    Hence, $C^\star$ is $2$-regular but has no simple cycle, so it must be a disjoint union of 
    (uncountably many) bi-infinite paths.    
\end{example}

\paragraph*{Class limits.}
We use symbols $\cl, \mathcal D,$ etc. to denote classes of finite graphs. We say that $\cl$ is {\em{hereditary}} if it is closed under taking induced subgraphs.
By $\gr$ we denote the class of all finite graphs.
For a class $\cl$ of finite graphs, we define the {\em{class limit}} $\cls$ to be the class comprising all subgraphs of ultraproducts of sequences of graphs from $\cl$.
For instance, $\gr^\star$ is the class of all graphs with at most $\cont$ vertices.

%\begin{lemma}\label{lem:}
%  If $\cl\subset \mathcal D$, then $\cls\subset \mathcal D^*$.
%\end{lemma}
\paragraph*{Shallow topological minors.}
We say that $H$ is a \emph{topological shallow minor of $G$ at depth $d$}
if there exists a mapping $\widetilde{\alpha}$, called the {\em{$d$-shallow topological minor model}}, which maps vertices of $H$ to distinct vertices of $G$, and edges of $H$ to paths in $G$, such that the following conditions hold:
\begin{itemize}
\item for each edge $uv\in E(H)$, $\widetilde{\alpha}(uv)$ is a path of length at most $2d+1$ with endpoints $\widetilde{\alpha}(u)$ and $\widetilde{\alpha}(v)$;
\item paths $\{\widetilde{\alpha}(e)\}_{e\in E(H)}$ are pairwise vertex-disjoint, apart from possibly sharing endpoints.
\end{itemize} 

By $\cl \topnab d $ we denote the class of all graphs $H$ which are topological shallow minors at depth $d$ of some graph in $\cl$, and for a single graph $G$ we denote $G\topnab d=\set{G}\nab d$. In particular, $\cl\topnab 0$ is the class of subgraphs of graphs from $\cl$.
We denote $\cl\topnab\omega=\bigcup_{d\in\Nat}\cl\topnab d$.

\paragraph*{Shallow  minors.}

%The \emph{topological realisation} of a graph $G$ with $m$ edges
%is the metric space  obtained by taking a disjoint union of $m$ unit intervals 
%and glueing them in the obvious way, according to their adjacency in $G$. The \emph{radius} of a graph $G$ is the smallest number $d$ such that there is a point $v$ in the topological realisation of $G$ which is at distance at most $d$ from every other point. This number is always a positive integer multiple of $1/2$.

Recall that the radius of a connected graph $G$ is the minimum integer $d$ for which there exists $w\in V(G)$ such that each vertex of $G$ is at distance at most $d$ from $w$.

Let $G,H$ be two (possibly infinite) graphs, and let $d\in\Nat$.
We say that $H$ is a \emph{shallow minor of $G$ at depth $d$}
if there exists a mapping $\alpha$ from vertices of $H$ to subgraphs of $G$, called the {\em{$d$-shallow minor model}}, that satisfies the following properties:
\begin{itemize}
\item each $\alpha(v)$ is a connected subgraph of $G$ of radius at most $d$;
% , for each $v\in V(H)$;
\item subgraphs $\alpha(u)$ and $\alpha(v)$ are disjoint for all distinct $u,v\in V(H)$;
\item for each $uv\in E(H)$, there is an edge in $G$ that connects a vertex of $\alpha(u)$ with a vertex of $\alpha(v)$.
\end{itemize}
The subgraphs ${\alpha(u)}$  of $G$, where ${u\in V(H)}$, are called the {\em{branch sets}} of $\alpha$. Since each branch set is connected and has radius at most $d$, for every $u\in V(H)$ we can select a {\em{center}} $\gamma(u)\in V(\alpha(u))$ such that each vertex of $\alpha(u)$ is at distance at most $d$ from $\gamma(u)$ in $\alpha(u)$.
% of radius at most $d$, such that for all vertices $v,w\in V(H)$, the graphs $\alpha(v)$ and $\alpha(w)$ are
%disjoint, and for adjacent vertices $v,w\in V(H)$, there 
%is an edge in $G$ between a vertex of $\alpha(v)$ and a vertex of $\alpha(w)$.

By $\cl \nab d $ we denote the class of all graphs $H$ which are shallow minors at depth $d$ of some graph in $\cl$. Similarly as before, $\cl\nab\omega=\bigcup_{d\in\Nat}\cl\nab d$ and $G\nab d=\set{G}\nab d$ for a single graph $G$. It is easy to see that if $H$ is a topological shallow minor of $G$ at depth $d$, then $H$ is also a shallow minor of $G$ at depth $d$, and hence $\cl\topnab d\subseteq \cl\nab d$.

The following simple result follows directly from the definition of a shallow minor by a simple calculation of distances.

\begin{lemma}[cf. Proposition 4.1 of~\cite{sparsity}]\label{lem:minor-of-my-minor}
If $G$ is a graph and $r,s\in \Nat$, then $(G\nab r)\nab s\subseteq G\nab (2rs+r+s)$.
\end{lemma}

We say that a $d$-shallow minor model $\tau$ of a graph $H$ in a graph $G$ is a \emph{tree} $d$-shallow minor model, if $\tau(u)$ is a tree, for each vertex $v$ of $H$.

\begin{lemma}\label{lem:branch-trees}
  If $H$ is a shallow minor of $G$ at depth $d$, then there is a tree $d$-shallow minor model $\tau$ of $H$ in $G$ whose branch sets are trees with at most $1+\Delta\cdot ({d-1})$ vertices, where $\Delta$ is the maximal degree of a vertex in $H$.
\end{lemma}

\begin{proof}
  Assume that $\alpha$ is any $d$-shallow minor model of $H$ in $G$.
  For an edge $vw$ in $H$,
   let $\eps(vw)$ be an edge in $G$ which connects a vertex in $\alpha(v)$ with a vertex in $\alpha(w)$. Let $W$ denote the set of all endpoints of edges of the form $\eps(vw)$.
Observe that $W$ has at most $\deg_H(v)$ nodes in $\alpha(v)$, for each $v\in V(H)$.
   
  Pick a vertex $v\in V(H)$. The branch set $\alpha(v)$ contains a spanning tree $\beta(v)$ whose radius and center $\gamma(v)$ are the same as in $\alpha(v)$. 
  Let $\tau(v)$ be the subtree of $\beta(v)$ induced by those vertices that lie on the shortest path from $\gamma(v)$ to some node in $W$.
  Then $\tau(v)$  has at most $\deg_H(v)\cdot (d-1)$ nodes different from $\gamma(v)$, and it is easy to see that by
    defining $\tau(v)$ as described above for each vertex $v\in V(H)$, we obtain a $d$-shallow minor model $\tau$ of $H$ in $G$.
\end{proof}

%\begin{lemma}\label{lem:}
%  If $H$ is a topological shallow minor of $G$ at depth $d$, then $H$ is a shallow minor of $G$ at depth $d$.
%\end{lemma}

\begin{lemma}\label{lem:shallow-fo}
Fix a finite graph $H$ and $d\in\Nat$.
\begin{itemize}
\item[(a)] There is a first order sentence $\phi_{H,d}$ such that $H\in G\nab d$ iff $G\models \phi_{H,d}$, for every graph $G$.
\item[(b)] There is a first order sentence $\tilde{\phi}_{H,d}$ such that $H\in G\topnab d$ iff $G\models \tilde{\phi}_{H,d}$, for every graph $G$.
\end{itemize}
\end{lemma}
\begin{proof}
% Suppose $\alpha$ is a minor model of $H$ in $G$, and let $\gamma(u)$ be a center of $\alpha(u)$ for each $u\in V(H)$. Construct an alternative minor model $\alpha'$ as follows: Take any edge $uv\in E(H)$, and any $xy\in E(G)$ such that $x\in \alpha(u)$ and $y\in \alpha(v)$. Construct a path $P_{uv}$ by concatenating (i) a path of length at most $d$ within $\alpha(u)$ that connects $\gamma(u)$ with $x$, (ii) the edge $xy$, and (iii) a path of length at most $d$ within $\alpha(v)$ that connects $y$ with $\gamma(v)$. Perform this construction for each $uv\in E(H)$. For each $u\in V(H)$, remove from $\alpha(u)$ all the vertices that are different from $\gamma(u)$ and do not lie on any path $P_{uv}$. Let $\alpha'(u)$ be the obtained modified branch set of $u$.

% From the construction it immediately follows that $\alpha'(u)$ is still a $d$-shallow minor model of $H$ in $G$. However, observe that each branch set $\alpha'(H)$ has size at most $1+(|V(H)|-1)d$, so the total number of vertices that appear in the branch sets $\set{\alpha'(u)}_{u\in V(H)}$ is at most a constant depending on $d$ and $|V(H)|$.

From Lemma~\ref{lem:branch-trees} it follows that if $G$ admits $H$ as a shallow minor at depth $d$, then also some constant-size subgraph of $G$ admits $H$ as a shallow minor at depth $d$, where the constant depends on the size of $H$ and $d$.
Hence, for (a) it suffices to verify the existence of a constant-size subset of vertices that induces a subgraph admitting $H$ as a shallow minor at depth $d$; this can be clearly done using a first order sentence. For (b), the conclusion that one only needs to look at constant-size induced subgraphs of $G$ follows directly from the definition of a shallow topological minor.
\end{proof}

\begin{lemma}\label{lem:ultrainclusion}
Fix a class of finite graphs $\cl$ and $d\in \Nat$. Then the following holds:
\begin{itemize}
\item[(a)] $(\cl\nab d)^\star\subseteq \cl^\star\nab d$, and
\item[(b)] $(\cl\topnab d)^\star\subseteq \cl^\star\topnab d$.
\end{itemize}
\end{lemma}
\begin{proof}
\noindent (a) Since $\cl^\star\nab d$ is closed under taking subgraphs, it suffices to show that every $H^\star\in (\cl\nab d)^\star$ of the form $H^\star=\ultraprod H_n$ for some sequence $(H_n)_{n\in \Nat}$ of graphs from $\cl\nab d$, belongs also to $\cl^\star\nab d$. For each $H_n$ we can find some $G_n\in \cl$ such that $H_n$ is a shallow minor of $G_n$ at depth $d$. Consider $G^\star=\ultraprod G_n$; then $G^\star\in \cl^\star$. It suffices to prove that $H^\star$ is a shallow minor of $G^\star$ at depth $d$. The idea of the following proof roughly amounts to observing that the property ``$\alpha$ is a $d$-shallow minor model of $H$ in $G$'' can be encoded in a relational structure encompassing $H,G,\alpha$. One then concludes by applying Łoś's theorem. The details follow.

Let $\alpha$ be a $d$-shallow minor model of some graph $H$ in a graph $G$.
Denote by $\Gamma_\alpha\subset V(H)\times V(G)$
 the binary relation
which relates $v\in V(H)$ with every $w\in V(\alpha(v))$.
Define a relational structure $\str M(H,G,\alpha)$
over a relational language $\Sigma$ consisting of two unary symbols, $V_H,V_G$,
and three binary symbols $E_H,E_G,\Gamma$, 
whose interpretations are $V(H),V(G),E(H),E(G),\Gamma_\alpha$, respectively. The universe of $\str M(H,G,\alpha)$ is the disjoint union of $V(H)$ and $V(G)$.

\begin{claim}\label{cl1}
  There is a first order sentence $\phi$ expressing the fact that 
  a structure $\str M$ over the language $\Sigma$,
with  relations $V_H,V_G,E_H,E_G,\Gamma$,
  is equal to $\str M(\alpha)$ for some $d$-shallow minor model $\alpha$ of the graph $(V_H,E_H)$ in the graph $(V_G,E_G)$.
\end{claim}
\begin{proof}
The sentence $\phi$ is a conjunction of the following conditions, each of which can be expressed by a first order sentence:
\begin{itemize}
	\item $V_H$ and $V_G$ form a partition of the universe, $E_H\subseteq V_H\times V_H$, $E_G\subseteq V_G\times V_G$, and both $E_H$ and $E_G$ are symmetric and anti-reflexive;

	\item $\Gamma$ is a relation with domain contained in $V_H$ and range contained in $V_G$ with the property that $(v,w)\in\Gamma$ and $(v',w)\in\Gamma$ implies $v=v'$, for every $v,v'\in V_H$ and $w\in V_G$;

	\item For every $v\in V_H$, the subgraph of the graph $(V_G,E_G)$ induced by $\Gamma(\set{v})=\set{w\in V_G: (v,w)\in \Gamma}$ is nonempty and has radius at most $d$;	
			
	\item Whenever $(v,w)\in E_H$, there exist elements $(v',w')\in E_G$ such that $(v,v'),(w,w')\in \Gamma$.
\end{itemize}
\cqed\end{proof}

We now come back to the proof of the lemma.
For $n\in \Nat$, let $\alpha_n$ be a  $d$-shallow minor model of $H_n$ in $G_n$, and let $\str M_n$ denote the structure $\str M(H_n,G_n,\alpha_n)$. Define $\str M=\ultraprod \str M_n$, and let the relations of $\str M$ be $V_H,V_G,E_H,E_G$ and $\Gamma$.
 Observe that $(V_H,E_H)$ is equal to $H^\star$ and $(V_G,E_G)$ is equal to~$G^\star$.
Since each structure $\str M_n$ satisfies the sentence $\phi$ from Claim~\ref{cl1}, by  Łoś's theorem, the ultraproduct $\str M$ also satisfies $\phi$. 
% Hence $\str M$ is equal to $\str M(\alpha)$, for some $d$-shallow minor model
% of  $H^\star$ in $G^\star$.
In particular, $H^\star$ is a shallow minor of $G^\star$ at depth $d$.

\medskip

\noindent (b) The proof follows exactly the same idea. This time, we encode a depth $d$-shallow topological minor model $\tilde\alpha$ as a structure $(V(H),V(G),E(H),E(G),\Gamma_{\tilde\alpha})$, where $\Gamma_{\tilde\alpha}\subset V(H)\times V(H)\times V(G)$ consists of triples $(v_1,v_2,w)$ such that $w\in V(\tilde\alpha(v_1v_2))$.
We then repeat the above argument. We leave the details to the reader.
\end{proof}

\begin{corollary}
If $\cl$ is a class of finite graphs, then $(\cl\nab 0)^\star=\cls$.
\end{corollary}
\begin{proof}
By Lemma~\ref{lem:ultrainclusion} and the definition of $\cls$ we have that $\cls\subseteq (\cl\nab 0)^\star \subseteq \cls\nab 0 =\cls$.
\end{proof}

%For (b), we first quantify existentially the distinct images of the vertices of $H$, and then verify the existence of the images of the edges of $H$. For $e=uv\in E(H)$, we say that there exist vertices $x^e_0,x^e_1,\ldots,x^e_{2d+1}$ such that $x^e_0$ is the image of $u$, $x^e_{2d+1}$ is the image of $v$, and for each $i=0,1,\ldots,2d$ we have that either $x^e_i=x^e_{i+1}$, or $x^e_i$ and $x^e_{i+1}$ are adjacent. Finally, we say that all the vertices $\set{x^e_i}_{e\in E(H),\, i=0,1,\ldots,2d+1$ are distinct whenever they are not the endpoints of the appropriate paths.

\section{Nowhere denseness}\label{sec:nowhere}
In this section we recall the standard definition of nowhere denseness for classes of finite graphs, and we develop its limit counterpart. That is, we introduce and prove the middle part of the diagram of Figure~\ref{fig:diagram}. 

Recall that a class $\cl$ of finite graphs is \emph{somewhere dense} if $\gr\subset\cl\nab d$ for some $d\in\Nat$, and is \emph{nowhere dense} otherwise. Equivalently, $\cl$ is nowhere dense if for every $d\in\Nat$ there exists $m\in\Nat$ such that $K_m\not\in\cl\nab d$.

\begin{lemma}\label{lem:dense}
Let $\cl$ be a class of finite graphs.
Then the following conditions are equivalent:  
\begin{enumerate}
\item\label{it:dense} $\cl$ is somewhere dense;
\item\label{it:cliques} $\gr \subseteq\cl\nab d$ for some $d\in\Nat$;
\item\label{it:ultracliques} $\gr \subseteq\cls\nab d$ for some $d\in\Nat$;
\item\label{it:omega} $K_\omega\in \cls\nab \omega$;
\item\label{it:cont} $K_\cont\in \cls\nab \omega$.
\end{enumerate}
\end{lemma}

\begin{proof}

The equivalence of the first two items is by definition. The implications~\implication{it:cont}{it:omega} and~\implication{it:omega}{it:ultracliques} are obvious. 
To prove~\implication{it:cliques}{it:cont}, suppose that for some $d\in\Nat$ we have that $\cl\nab d=\gr$. Then, by Lemma~\ref{lem:ultrainclusion} we have that $\cls\nab d\supseteq \gr^\star$. However, from Example~\ref{ex:cliques} it follows that $K_\cont\in \gr^\star$, so $K_\cont\in \cls\nab d\subseteq \cls\nab \omega$.
%  It follows that $\cls\nab d\supseteq \gr^*\nab 0$. Therefore, $\cls\nab d$
%  contains the clique with $2^{\aleph_0}$ vertices, as argued in Example~\ref{ex:cliques}.
  
It remains to prove implication~\implication{it:ultracliques}{it:cliques}. Suppose that $\gr \subseteq\cls\nab d$ for some $d\in\Nat$. Let us fix some finite graph $H$; then in particular $H\in \cls\nab d$. Hence, there exists a graph $G^\star\in \cls$ such that $H\in G^\star\nab d$. By the definition of $\cls$, we can assume that $G^\star=\ultraprod G_n$ for some sequence $(G_n)_{n\in \Nat}$ of graphs from $\cl$. By applying Łoś's theorem to formula $\phi_{H,d}$ given by Lemma~\ref{lem:shallow-fo}, we infer that property $H\in G_n\nab d$ holds \uae. Consequently, $H\in G_n\nab d$ for at least one index $n\in \Nat$, so also $H\in \cl\nab d$. As $H$ was chosen arbitrarily, we have that $\gr \subseteq\cl\nab d$.
\end{proof}

%Therefore, there is such a sequence $(G_n)_{n\in \Nat}$ of graphs in $\cl$ 
%  that $K_\omega$ is a shallow minor at depth $d$ of $\ultraprod G_n$.
%  In particular,  by Łoś'ss theorem and Lemma~\ref{lem:shallow-fo}, for every $m\in\Nat$, the property $K_m\in G_n\nab d$ holds \uae.
%Consequently, $K_m\in \cl\nab d$ for every $m\in\mathbb N$.

Let us call a class $\cl$ of finite graphs \emph{topologically somewhere dense} if $\gr\subset\cl\topnab d$ for some $d\in\Nat$, and \emph{topologically nowhere dense} otherwise. By applying the same proof, but using items (b) instead of (a) in Lemmas~\ref{lem:shallow-fo} and~\ref{lem:ultrainclusion}, we obtain the following result:

\begin{lemma}\label{lem:top-dense}
Let $\cl$ be a class of finite graphs.
Then the following conditions are equivalent:  
\begin{enumerate}
\item\label{it:dense} $\cl$ is topologically somewhere dense;
\item\label{it:cliques} $\gr \subseteq\cl\topnab d$ for some $d\in\Nat$;
\item\label{it:ultracliques} $\gr \subseteq\cls\topnab d$ for some $d\in\Nat$;
\item\label{it:omega} $K_\omega\in \cls\topnab \omega$,
\item\label{it:cont} $K_\cont\in \cls\topnab \omega$.
\end{enumerate}
\end{lemma}

Lemmas~\ref{lem:dense} and~\ref{lem:top-dense} suggest the following definitions of the limit counterparts of nowhere denseness: A class $\cl$ of finite graphs is {\em{limit somewhere dense}} if $K_\omega\in \cls\nab \omega$, and {\em{limit nowhere dense}} otherwise. Respectively, $\cl$ is {\em{limit topologically somewhere dense}} if $K_\omega\in \cls\topnab \omega$, and {\em{limit topologically nowhere dense}} otherwise.

It is a classic result of the theory of sparse graphs that nowhere denseness is equivalent to topological nowhere denseness~\cite{sparsity}. We can now give an alternative proof of this statement by showing that the limit counterparts are equivalent, and then using Lemmas~\ref{lem:dense} and~\ref{lem:top-dense} to transfer this equivalence to the standard definitions.

\begin{lemma}\label{lem:limit-top-equiv}
The following conditions are equivalent.
\begin{enumerate}
\item\label{omega:nabla} $\cl$ is limit somewhere dense, i.e., $K_\omega\in \cls\nab \omega$.
\item\label{omega:top} $\cl$ is limit topologically somewhere dense, i.e., $K_\omega\in \cls\topnab \omega$.
\end{enumerate}
\end{lemma}
\begin{proof} 
Implication \implication{omega:top}{omega:nabla} is obvious, since $\cls\topnab \omega\subseteq \cls\nab\omega$.
It remains to show implication \implication{omega:nabla}{omega:top}.
The following auxiliary claim will be useful.

%Below we consider two distinct countably infinite cliques $K$ and $K_\omega$, to avoid confusion. Suppose that each directed edge in $K$ is colored red or blue.
%For vertices $v,x,y,w$ of $K$, call the path $vxyw$ of length three in $K$ \emph{red-handed} if both directed edges $vx$ and $wy$ are colored red.
  
\begin{claim}\label{cl:greedy-model}
Suppose $K$ is an infinite clique, and suppose $R$ is a directed graph with $V(R)=V(K)$ such that in $R$ each vertex has an infinite outdegree.
Then there exists a $1$-shallow topological minor model $\widetilde\beta$ of $K_\omega$ in $K$ such that for each $pq\in E(K)$, $\widetilde\beta(pq)$ is {\em{compatible with $R$}}: it is a path $\widetilde\beta(p)-x-y-\widetilde\beta(q)$ of length three that satisfies $(\widetilde\beta(p),x)\in E(R)$ and $(\widetilde\beta(q),y)\in E(R)$.
\end{claim}

\begin{proof} 
  Choose any sequence $H_0\subset H_1\subset H_2\subset \ldots$
  of finite subgraphs of $K_\omega$ with the properties that $K_\omega=\bigcup_{n=0}^\omega H_n$, the graph $H_0$ is empty, and for each $n\ge 0$, $H_{n+1}$  is obtained from $H_n$ by either adding a vertex, or adding an edge.
  
  The model $\widetilde\beta$ of $K_\omega$ in $K$ is constructed inductively. At the $n$-th step of the induction we maintain the invariant that $\widetilde\beta$ is a $1$-shallow topological minor model of $H_n$ in a finite subgraph of $K$,
  such that  each path $\widetilde\beta(pq)$, for $pq\in E(H_n)$, is compatible with $R$.
   
In the $0$-th step, $\widetilde\beta$ is the empty mapping. 
In the inductive step, consider two cases. If $H_{n+1}$ is obtained from $H_n$
by adding a vertex $p$ of $K_\omega$, then extend $\widetilde\beta$ by setting $\widetilde\beta(p)=v$ for any vertex $v$ of $K$ which has not been used so far by $\widetilde\beta$. In the second case, $H_{n+1}$ is obtained from $H_n$ by adding an edge $pq$ between two vertices of $H_n$. In this case, let $v=\widetilde\beta(p)$ and $w=\widetilde\beta(q)$ and
choose any two distinct vertices $x,y$ of $K$ 
such that the $(v,x)\in E(R)$, $(w,y)\in E(R)$, and $x,y$ are not used by $\widetilde\beta$ constructed so far.
Such vertices exist, since $v$ has infinite outdegree in $R$ and the graph used by $\widetilde\beta$ so far is finite; likewise for $w$.
Extend $\widetilde\beta$ by setting $\widetilde\beta(pq)$ to the path $v-x-y-w$.
It is easy to see that the invariant is maintained. 

Eventually, for every pair of vertices $p,q$ of $K_\omega$, the vertices $\widetilde\beta(p)$ and $\widetilde\beta(q)$ are defined, as well as the path $\widetilde\beta(pq)$, which is compatible with $R$. Therefore, the mapping $\widetilde\beta$ obtained in this process satisfies the conditions of the~lemma.
\cqed\end{proof}

  We proceed with the proof of implication \implication{omega:nabla}{omega:top}. To this end, it suffices to show that 
  if some infinite clique $K$ is a $d$-shallow minor of a graph $G$, for some $d\in \Nat$, then $K_\omega$ is a topological shallow minor of $G$ at depth $3d+1$. 
  
  Let $\tau$ be a depth-$d$ tree minor model of $K$ in $G$; such model exists by Lemma~\ref{lem:branch-trees}. 
  For each vertex $v$ of $K$, let $\gamma(v)$ be an arbitrarily selected center of $\tau(v)$, and 
for each edge $uv$ of $K$, let $\eps(uv)$ be an edge in $G$ that
  connects a vertex in $\tau(u)$ with a vertex in $\tau(v)$.

We view each $\tau(u)$ as a rooted tree with root $\gamma(u)$, whose depth is at most $d$. Let $\tau'(u)$ be the rooted tree obtained by adding to the tree $\tau(u)$ 
the edges $\eps(uv)$ (and their endpoints), for all $v\in V(K)\setminus \{u\}$. Each of these edges is appended to a leaf of $\tau(u)$, and thus $\tau'(u)$ is a rooted tree of depth at most $d+1$. Since $\tau'(u)$ has an infinite number of leaves, there exists a vertex $\sigma(u)$ that has an infinite number of children in $\tau'(u)$. 

Construct a directed graph $R$ with $V(R)=V(K)$, where the edge set of $R$ is defined as follows. For every $u\in V(K)$, and every subtree of $\tau'(u)$ that is rooted at a child of $\sigma(u)$, pick an arbitrary vertex $v\in V(K)$ such that $\eps(uv)$ is contained in this subtree, and put $(u,v)$ into the edge set of $R$. Since $\sigma(u)$ has infinitely many children in $\tau'(u)$, it follows that in $R$ every vertex has an infinite outdegree. Observe that whenever $(u,v)\in E(R)$, one can find a path of length at most $d+1$ in $\tau'(u)$ that leads from $\sigma(u)$ downwards to $\eps(uv)$. Note that for a fixed $u$, all these paths will be vertex-disjoint apart from the endpoint $\sigma(u)$.

Apply Claim~\ref{cl:greedy-model} to $K$ and $R$, yielding a topological minor model $\widetilde\beta$ of $K_\omega$ in $K$. From this, we construct a topological model of $K_\omega$ in $G$, as follows.
For a vertex $p$ of $K_\omega$, set $\widetilde\alpha(p)=\sigma(\widetilde\beta(p))$. For an edge $pq$ of $K_\omega$, suppose that $\widetilde\beta(pq)$ is a path $v-x-y-w$ in $K$, where $v=\widetilde\beta(p)$, $w=\widetilde\beta(q)$, $(v,x)\in E(R)$ and $(w,y)\in E(R)$. Then, set $\widetilde\alpha(pq)$ to be the path from $\sigma(v)$ to $\sigma(w)$ defined as the union of three paths
$\lambda,\mu,\rho$ in $G$, where
\begin{itemize}
  \item $\lambda$ is the path in $\tau'(v)$ going downwards from $\sigma(v)$ 
  and ending with the edge $\eps(vx)$. The path exists since $(v,x)\in E(R)$. It has length at most $d+1$,
  and ends in a vertex~$l$~of~$\tau(x)$.
  \item $\rho$ is the path in $\tau'(w)$ going downwards from $\sigma(w)$ 
  and ending with the edge $\eps(wy)$. The path exists since $(w,y)\in E(R)$. It has length at most $d+1$,
  and ends in a vertex~$r$~of~$\tau(y)$.
  
  \item $\mu$ is any path in $G$ connecting $l$ with $r$,
  contained in $\tau(x)\cup\eps(xy)\cup \tau(y)$.
  We can choose this path so that its length is at most $2d+1+2d$.
\end{itemize}
Altogether, the path $\widetilde\alpha(pq)$ has length at most 
$6d+3$. It is easy to see that the mapping $\widetilde\alpha$ is a $(3d+1)$-shallow topological minor model of $K_\omega$ in $G$.
\end{proof}

\section{Quasi-wideness}\label{sec:quasi}
In this section, we recall the notions of (uniform) quasi-wideness  introduced in~\cite{Dawar10},
introduce their limit counterparts, and prove equivalence of the limit variants with limit nowhere-denseness.

%\begin{fact}
%  $A$ is $d$-scattered $\iff$ $A$ is $2d$-independent.
%\end{fact}

\subsection{Definitions}\label{sec:quasi-def}

Suppose $G$ is a graph and $d\in \Nat$. Recall that we say that a vertex subset $X\subseteq V(G)$ is {\em{$d$-scattered}} if for all distinct $u,v\in X$, the distance between $u$ and $v$ is larger than $2d$. Equivalently, balls $\set{N_G^d[u]}_{u\in X}$ are pairwise disjoint.
For the remainder of this section, let us fix some class $\cl$ of finite graphs. 

Let $s\colon \Nat\to\Nat$ be a function. We say that $\cl$ is \emph{quasi-wide (\emph{QW}) with margin $s$}  if for all $d,m\in \Nat$ there exists an $N\in \Nat$ such that the following holds: For every graph $G\in \cl$ with $|V(G)|\geq N$, there exists $S\subseteq V(G)$ with $|S|\leq s(d)$ such that $G-S$ contains a $d$-scattered set of size $m$. We say that $\cl$ is \emph{uniformly quasi-wide ({UQW}) with margin $s$} if for all $d,m\in \Nat$ there exists an $N\in \Nat$ such that the following holds: For every graph $G\in \cl$ and every vertex subset $W\subseteq V(G)$ with $|W|\geq N$, there exists $S\subseteq V(G)$ with $|S|\leq s(d)$ such that in $G-S$ one can find a $d$-scattered set of size $m$ contained in $W$. Note that uniform quasi-wideness implies quasi-wideness with the same margin by taking $W=V(G)$. 

It is a known result of the theory of sparse graphs (cf. Theorem 8.2 of~\cite{sparsity}) that a hereditary class $\cl$ of finite graphs is (uniformly) quasi-wide with some margin if and only if $\cl$ is nowhere dense. The goal of this section is to give an alternative proof of this statement by proving it for the limit counterparts of (uniform) quasi-wideness, and then transferring this result to the standard definitions by the means of Łoś's theorem. 

To this end, we now give limit counterparts of the notions of (uniform) quasi-wideness. There will be four variants: the definition will either contain the uniformity requirement or not, and the margin will be either uniformly bounded by some function $s$, or just required to be finite. As we will see later, for hereditary $\cl$ all these notions will be equivalent, and equivalent to (limit) nowhere denseness and standard (uniform) quasi-wideness.

We shall say that $\cl$ is \emph{limit quasi-wide} (\emph{LQW}) if for every $d\in \Nat$ and every infinite graph $G\in \cls$, there exists a finite set $S\subseteq V(G)$ such that in $G\setminus S$ one can find an infinite $d$-scattered set. We shall also say that $\cl$ is \emph{uniformly limit quasi-wide} (\emph{ULQW}) if for every $d\in \Nat$, every graph $G\in \cls$, and every infinite set $W\subseteq V(G)$, there exists a finite set $S\subseteq V(G)$ such that in $G\setminus S$ one can find an infinite $d$-scattered set contained in $W$.

Let $s\colon \Nat\to \Nat$ be a function. We shall say that $\cl$ is \emph{strongly limit quasi-wide ({SLQW}) with margin $s$}  if for every $d\in \Nat$ and every infinite graph $G\in \cls$, there exists a set $S\subseteq V(G)$ with $|S|\leq s(d)$ such that in $G\setminus S$ one can find an infinite $d$-scattered set. We shall also say that $\cl$ is \emph{strongly uniformly limit quasi-wide ({SULQW}) with margin $s$}  if for every $d\in \Nat$, every graph $G\in \cls$, and every infinite set $W\subseteq V(G)$, there exists a set $S\subseteq V(G)$ with $|S|\leq s(d)$ such that in $G\setminus S$ one can find an infinite $d$-scattered set contained in $W$.

The following implications are trivial: If $\cl$ is strongly limit quasi-wide with some margin, then it is also limit quasi wide. The same holds for the uniform variants: SULQW implies UQW. By taking $W=V(G)$, we obtain that uniform limit quasi-wideness implies limit quasi-wideness, and strong uniform limit quasi-wideness with some margin $s$ implies strong limit quasi-wideness with the same margin.

\subsection{Linking standard and limit notions}

We now link the standard notions of (uniform) quasi-wideness with their limit counterparts. For technical reasons, we can prove this connection directly only for the strong variants of the notion. This technical caveat is the main reason why we introduced the strong variants in the first place.

\begin{lemma}\label{lem:quasi-los}
Let $s\colon \Nat\to \Nat$ be a function. Then the following implications hold:
\begin{itemize}
\item[(a)] If $\cl$ is strongly limit quasi-wide with margin $s$, then $\cl$ is quasi-wide with margin $s$.
\item[(b)] If $\cl$ is strongly uniformly limit quasi-wide with margin $s$, then $\cl$ is uniformly quasi-wide with margin $s$.
\end{itemize}
\end{lemma}
\begin{proof}
We first prove (b), and then discuss how to modify the reasoning to obtain a proof of (a).

For the sake of contradiction, suppose $\cl$ is not uniformly quasi-wide with margin $s$. This means that there are some $d,m\in \Nat$, a sequence of graphs $(G_n)_{n\in \Nat}$, and a sequence of sets $(W_n)_{n\in \Nat}$, where $W_n\subseteq V(G_n)$ for each $n\in \Nat$, such that the following holds for every $n\in \Nat$:
\begin{itemize}
\item $|W_n|\geq n$, and
\item for any set $S\subseteq V(G_n)$ with $|S|\leq s(d)$, in $G_n\setminus S$ there is no $d$-scattered set $X$ with $X\subseteq W_n$ and $|X|=m$.
\end{itemize}
Let $G^\star=\ultraprod G_n$ and $W^\star=\ultraprod W_n$. Then $G^\star\in \cls$, and the same reasoning as in Example~\ref{ex:cliques} shows that $W^\star$ is infinite.

Consider a first order formula $\phi$ in the language of graphs with a specified vertex subset $W$ (encoded as a unary relation) that says the following: ``There are vertices $u_1,u_2,\ldots,u_{s(d)},v_1,v_2,\ldots,v_m$ such that (i) all $\set{v_i}_{1\leq i\leq m}$ are pairwise different, different from all $\set{u_j}_{1\leq j\leq s(d)}$, and contained in $W$; and (ii) for all $1\leq i<i'\leq m$ there is no path of length at most $2d$ between $v_i$ and $v_{i'}$ that would avoid all the vertices $\set{u_j}_{1\leq j\leq s(d)}$''. By our assumption about $G_n$ and $W_n$, we have that $G_n,W_n\nmodels \phi$ for all $n\in \Nat$. However, since $\cl$ is strongly uniformly limit quasi-wide with margin $s$, and $W^\star$ is infinite, there is a set $S\subseteq G^\star$ with $|S|\leq s(d)$ such that in $G^\star\setminus S$ there is an infinite $d$-scattered set contained in $W^\star$. In particular there is one of size $m$, so $G^\star,W^\star\models \phi$. This is a contradiction with Łoś's theorem.

The proof of (a) follows the same lines, but we do not need to consider sets $\set{W_n}_{n\in \Nat}$ and $W^\star$. Consequently, $\phi$ is simply in the language of graphs and does not require $\set{v_i}_{1\leq i\leq m}$ to be contained in $W$. The rest of the reasoning is the same.
\end{proof}

Lemma~\ref{lem:quasi-los} provides only one direction of implications: from limit variants to the standard ones. Unfortunately, it seems difficult to directly prove the reverse implications. This is because the definition of (uniform) limit quasi-wideness is not hereditary. More precisely, in order to prove that $\cl$ is strongly (uniformly) limit quasi-wide, we need to consider an arbitrary subgraph $G^\star$ of an ultraproduct $\ultraprod G_n$ of graphs from $\cl$. In case of the uniformity condition, also an arbitrary subset $W^\star$ of its vertices needs to be considered. Both $G^\star$ and $W^\star$ may not be defined as ultraproducts themselves, while the existence of large scattered sets in $\ultraprod G_n$ can have nothing to do with the existence of large scattered sets in $G^\star$. Therefore, it seems difficult to use the assumed quasi-wideness of $\cl$ to reason about $G^\star$ and $W^\star$.

For this reason, we resort to invoking an implication known from the literature.

\begin{lemma}[Lemma 8.9 and Theorem 8.2 of~\cite{sparsity}]\label{lem:finite-qw}
Let $\cl$ be a hereditary class of finite graphs. If $\cl$ is quasi-wide with some margin $s$, then $\cl$ is also topologically nowhere dense.
\end{lemma}

Note that Lemma~\ref{lem:finite-qw} provides the easy implication of the equivalence between quasi-wideness and nowhere denseness. Its proof essentially boils down to observing that a sufficiently large clique with each edge subdivided at most $2d$ times is a counterexample for quasi-wideness. For the harder implication --- from nowhere denseness to quasi-wideness --- in the next section we provide a proof for the limit variants. By combining this with Lemmas~\ref{lem:dense} and~\ref{lem:quasi-los}, we can infer the same implication for the standard variants.

We note that the assumption that $\cl$ is hereditary is necessary in Lemma~\ref{lem:finite-qw}. This is not the case for the implications starting with the limit variants of quasi-wideness, because we explicitly defined $\cls$ to be closed under taking subgraphs.

\subsection{Quasi-wideness and nowhere denseness}

We now link the limit variants of quasi-wideness and nowhere denseness. First we prove the easier direction.

\begin{lemma}\label{lem:qw-tnd}
If $\cl$ is limit quasi-wide, then $\cl$ is limit topologically nowhere dense.
\end{lemma}
\begin{proof}
  To prove the contrapositive, suppose that for some $d\in \Nat$ it holds that $K_\omega\in \cls\topnab d$. We show that $\cl$ is not quasi-wide.
  
  Since $\cls$ is closed under taking subgraphs by definition, there is a graph $H\in \cls$ that can be obtained from $K_\omega$ by subdividing each edge at most $2d$ times. Let $A$ be the set of vertices of infinite degree in $H$. Let $S\subseteq V(H)$ be any finite subset of vertices. Observe that every two vertices of $A$ can be connected by an infinite family of internally vertex-disjoint paths in $H$ of length at most $4d+2$. As $S$ is finite, it follows that all the vertices of $A\setminus S$ belong to the same connected component $C_0$ of $H\setminus S$, and moreover they are pairwise at distance at most $4d+2$ within $C_0$. Since each path connecting a pair of vertices from $A$ has length at most $2d+1$, we infer that the diameter of $C_0$ is at most $8d+2$. Each connected component of $H\setminus S$ apart from $C_0$ is contained in the internal vertices of some path connecting a pair of vertices of $A$, and hence has size at most $2d-1$. Moreover, there can be at most $\binom{|S|}{2}$ such components. We conclude that $H\setminus S$ has a finite number of connected components, each of diameter at most $8d+2$. As $S$ was chosen arbitrarily,  $\cl$ cannot be limit quasi-wide.
\end{proof}

%  We prove the contrapositive. Suppose that $K_\omega\in \cls\nab k$, for some finite $k$. Let $G\in\cls$ be such that $K\in G\nab k$.  
%  Then, for no finite $S\subset V(G)$,
%  $G-S$ contains an infinite $r$-scattered set, so $\cl$
%  is not limit quasi-wide.

We proceed to the more difficult direction. The proof of this implication is inspired by the proof of the analogus statement for the standard nowhere-denseness and uniform quasi-wideness; see Section 8.3 of~\cite{sparsity} for an exposition.

We first need an auxiliary lemma.

\begin{lemma}\label{lem:aux-big}
Suppose $H$ is a bipartite graph with $A$ and $B$ being the sides of the bipartition. Suppose further that $A$ is infinite and that every vertex of $B$ has a finite degree. Then one of the following conditions holds:
\begin{itemize}
\item[(a)] There exists an infinite set $X\subseteq A$ such that $N(u)\cap N(v)=\emptyset$ for all distinct $u,v\in X$.
\item[(b)] $H$ admits $K_\omega$ as a $1$-shallow minor.
\end{itemize}
\end{lemma}
\begin{proof}
Color each pair $\set{u,v}\subseteq A$ white if $N(u)\cap N(v)=\emptyset$, and black otherwise. From Ramsey's theorem it follows that there is either a white or a black infinite clique, i.e., an infinite subset $X\subseteq A$ such that either all the pairs of vertices from $X$ are white, or all of them are black. If all of them are white, then $X$ satisfies condition  (a). Hence, suppose that for all distinct $u,v\in X$, we have that $N(u)\cap N(v)\neq \emptyset$. We prove that then $H$ admits $K_\omega$ as a $1$-shallow minor.

We construct a minor model $\alpha$ of $K_\omega$ by induction, starting with the empty model. At step $n$ we define $\alpha(n)$ so that the already defined branch sets $\alpha(m)$, for $0\leq m\leq n$, form a $1$-shallow minor model of $K_{n+1}$. Moreover, we shall keep the invariant that each $\alpha(m)$ consists of exactly one vertex $u_m$ from $X$ and a finite number of vertices from $B$, all adjacent to $u_m$. Thus, it will be clear that $\alpha(m)$ has radius at most $1$.

Suppose that we have performed steps $0,1,\ldots,n-1$, and now we need to define $\alpha(n)$ so that it is adjacent to each $\alpha(m)$ for $m<n$. Let $Z=\bigcup_{m=0}^{n-1} V(\alpha(m))$, $Z_A=Z\cap A$ and $Z_B=Z\cap B$. From the invariant it follows that $|Z_A|=n$ and $Z_B$ is finite. Let $Y=Z_A\cup N(Z_B)$. Since the vertices of $B$ have finite degrees in $H$, it follows that $Y$ is finite. Set $u_n$ to be an arbitrarily chosen vertex of $X\setminus Y$; such a vertex exists because $X$ is infinite and $Y$ is finite. Take any $m$ with $0\leq m<n$. Since $u_m,u_n\in X$, there exists a vertex $b_{m,n}\in B$ that is a common neighbor of $u_m$ and $u_n$. Moreover, $b_{m,n}$ is not yet used in any branch set $\alpha(m')$ for $m'<n$, because then we would have that $b_{m,n}\in Z_B$ and $u_n\in N(Z_B)\subseteq Y$. Therefore, we can set $\alpha(n)=\set{u_n}\cup \set{b_{m,n}\colon 0\leq m<n}$, which satisfies the required invariant.
\end{proof}

\begin{lemma}\label{lem:big-proof}
If $\cl$ is limit nowhere dense, then $\cl$ is strongly uniformly limit quasi-wide with some margin $s$.
\end{lemma}
\begin{proof}
By Lemma~\ref{lem:dense}, the assumption is equivalent to the statement that $\gr\nsubseteq \cls\nab d$ for every $d\in \Nat$. Let then $\kappa\colon \Nat\to \Nat$ be a function such that $K_{\kappa(d)+1}\notin \cls \nab d$, for all $d\in \Nat$.

Take any graph $G\in \cls$ and suppose $W\subseteq V(G)$ is infinite. We inductively define sets $\emptyset=S_0\subseteq S_1\subseteq S_2\subseteq \ldots$ and $W=X_0\supseteq X_1\supseteq X_2\supseteq \ldots$ with the following properties:
\begin{itemize}
\item $|S_{d+1}\setminus S_{d}|\leq \kappa(3d+1)$, for all $d\in \Nat$.
\item Each $X_d$ is infinite, disjoint with $S_d$, and $d$-scattered in $G\setminus S_d$.
\end{itemize}
As $G$ and $W$ are chosen arbitrarily, the existence of such sets $(S_d)_{d\in \Nat}$ and $(X_d)_{d\in \Nat}$ shows that $\cl$ is strongly uniformly limit quasi-wide with margin $s(d)=\sum_{i=0}^{d-1} \kappa(3i+1)$. For the base of the induction, we can take $S_0=\emptyset$ and $X_0=W$.

Take $d\in \Nat$ and suppose that $S_{d}$ and $X_{d}$ are defined. Let $G'=G\setminus S_{d}$; then $X_{d}$ is $d$-scattered in $G'$. For $u\in X_{d}$, consider the ball $N_u=N^d_{G'}[u]$. Then the balls $\set{N_u}_{u\in X_{d}}$ are pairwise disjoint. For every pair of distinct vertices $\set{u,v}\subseteq X_{d}$, color this pair black if $N_u$ and $N_v$ are connected by an edge, and white otherwise. Since $X_{d}$ is infinite, from Ramsey's theorem it follows that there is either a white or a black infinite clique, i.e., an infinite subset $X'\subseteq X_{d}$ such that either all the pairs of distinct vertices of $X'$ are colored white, or all of them are colored black. If they were all black, then mapping $u\to G[N_u]$ for $u\in X'$ would be a $d$-shallow minor model of an infinite clique in $G$, a contradiction with the assumption that $\cl$ is limit nowhere dense. Hence, for every pair of distinct vertices $u,v\in X'$, there is no edge between the balls $N_u$ and $N_v$.

Let us construct a bipartite graph $H$ as follows: Take $A=X'$ and $B=V(G')\setminus \bigcup_{u\in X'} N_u$ to be the sides of the bipartition. For $u\in A$ and $v\in B$, put $uv\in E(H)$ iff $v$ has a neighbor in $N_u$. As each ball $N_u$ for $u\in X'$ has radius at most $d$, it follows that $H$ is a $d$-shallow minor of $G$.

Consider the following procedure performed in $H$. Start with $Y=A$ and $R=\emptyset$, and as long as there exists a vertex $w\in B\setminus R$ that has an infinite number of neighbors in $Y$, add $w$ to $R$ and replace $Y$ with $Y\cap N_H(w)$. Thus, $Y$ stays infinite throughout the whole procedure. We claim that the procedure stops after at most $\kappa(3d+1)$ steps. Indeed, otherwise at step $\kappa(3d+1)+1$ we would construct a set $R\subseteq B$ with $|R|=\kappa(3d+1)+1$ and an infinite set $Y\subseteq A$ such that $H[R\cup Y]$ is a biclique. In this biclique, we can find $K_{\kappa(3d+1)+1}$ as a $1$-shallow minor. Since $H$ itself is a $d$-shallow minor of $G$, from Lemma~\ref{lem:minor-of-my-minor} we would obtain that $K_{\kappa(3d+1)+1}$ is a $(3d+1)$-shallow minor of $G$, a contradiction with the definition of $\kappa(3d+1)$. 

Therefore, at the end of the procedure we obtain a set $R\subseteq B$ with $|R|\leq \kappa(3d+1)$ and an infinite set $Y\subseteq A$ such that each vertex of $B\setminus R$ has only a finite number of neighbors in $Y$. Consider $H'=H[Y\cup (B\setminus R)]$. Then $H'$ satisfies the prerequisites of Lemma~\ref{lem:aux-big}. Moreover, $H'$ cannot have $K_\omega$ as a $1$-shallow minor, because then, by Lemma~\ref{lem:minor-of-my-minor}, $G$ would admit $K_\omega$ as a $(3d+1)$-shallow minor, a contradiction with $\cl$ being limit nowhere dense. 

Hence, from Lemma~\ref{lem:aux-big} we infer that there is an infinite set $Z\subseteq Y$ such that vertices of $Z$ pairwise do not have common neighbors in $H'$. As $Y\subseteq X'$, this means that in $G'\setminus R$ the balls of radius $d$ around vertices of $Z$ are pairwise disjoint, pairwise non-adjacent, and they pairwise do not have common neighbors. We infer that in $G'\setminus R$ the balls of radius $d+1$ around vertices of $Z$ are pairwise disjoint; equivalently, $Z$ is $(d+1)$-scattered in $G'\setminus R$. Hence we can take $X_{d+1}=Z$ and $S_{d+1}=S_d\cup R$. This concludes the induction step.
\end{proof}

We remark that if one is only interested in the implication from limit nowhere denseness to uniform limit quasi-wideness (i.e., without the strongness condition), then there is no need of using the slightly stronger assumption provided by Lemma~\ref{lem:dense} --- the assumption $K_\omega\notin \cls\nab \omega$ suffices. Namely, in this setting, when constructing $R$ we only need to argue that the procedure stops after a finite number of steps. If the procedure could be performed indefinitely, then it is easy to expose a $1$-shallow minor of $K_\omega$ in $H$.

\section{Splitter game}\label{sec:game}
In this section we define and study the limit variant of the splitter game, introduced by Grohe et al.~\cite{GroheKS14} in the context of proving fixed-parameter tractability of model checking first order logic on nowhere dense graph classes.

\subsection{Definitions}

We first recall the definition of the splitter game due to Grohe et al.~\cite{GroheKS14}. Fix some graph $G$ and $\ell,m,d\in \Nat$. The {\em{$(\ell,m,d)$-splitter game}} on $G$ is a game played by two players: connector and splitter. The game proceeds in $\ell$ rounds. At the beginning of round $i+1$, the game is played on a subgraph $G_i$ of $G$ that we shall call the {\em{arena}}; the game starts with $G_0=G$. In round $i+1$, connector first picks an arbitrary vertex $v_{i+1}\in V(G_i)$ of the arena. Then, splitter chooses a subset $W_{i+1}\subseteq V(G_i)$ of size at most $m$. The arena gets trimmed to $G_{i+1}=G_i[N^d_{G_i}[v_{i+1}]\setminus W_{i+1}]$; that is, in the next round we consider only vertices that are at distance at most $d$ in $G_i$ from the vertex picked by connector, and moreover the vertices $W_{i+1}$ chosen by splitter are removed. Splitter wins the game if within $\ell$ rounds the arena becomes empty, i.e., $G_{i+1}$ is an empty graph for some $i\leq \ell$, and connector wins otherwise. 

A {\em{winning strategy}} for connector is simply a function $\sigma$ that maps sequences of the form $(v_1,W_1,v_2,W_2,\ldots,v_i,W_i)$ to vertices of $G$ such that (i) $\sigma(v_1,W_1,v_2,W_2,\ldots,v_i,W_i)$ is a valid move $v_{i+1}$ for connector; and (ii) by playing $v_{i+1}=\sigma(v_1,W_1,v_2,W_2,\ldots,v_i,W_i)$ in each round, where $v_1,v_2,\ldots,v_i$ and $W_1,W_2,\ldots,W_i$ are previous moves of connector and splitter, respectively, connector wins the $(\ell,m,d)$-splitter game. We similarly define winning strategies for splitter as functions $\tau$ the map sequences of the form $(v_1,W_1,v_2,W_2,\ldots,v_i,W_i,v_{i+1})$ to sets $W_{i+1}\subseteq V(G)$.

We will say that a class $\cl$ of finite graphs is {\em{winning for splitter}} if for every $d\in \Nat$, there exist $\ell,m\in \Nat$ such that splitter has a winning strategy in the $(\ell,m,d)$-splitter game on every graph from $G$. Grohe et al.~\cite{GroheKS14} proved that $\cl$ is winning for splitter if and only if it is nowhere dense. Our goal now is to give an alternative proof of this result by defining the limit variant of the splitter game, proving its equivalence with the standard variant by the means of Łoś's theorem, and then showing that splitter wins in the limit game if and only if the graph class is limit nowhere-dense.

Fix some graph $G$ and $d\in \Nat$. The {\em{limit $d$-splitter game}} has the same rules as the standard $(\ell,m,d)$-splitter game apart from the following differences:
\begin{itemize}
\item In splitter's moves, we require that $W_{i+1}$ is finite (instead of bounding its size by $m$).
\item Splitter wins if the arena becomes empty after a finite number of rounds (instead of after at most $\ell$ rounds).
\end{itemize}
The definitions of strategies for splitter and connector are the same as previously. We say that a class $\cl$ of finite graphs is {\em{limit winning for splitter}} if for every $d\in \Nat$, splitter has a winning strategy in the limit $d$-splitter game on every graph from $\cls$.

The following simple observation shows that in fact we can restrict ourselves to the game, where splitter always pick singleton sets $W_i$.

\begin{lemma}\label{lem:one-move}
Let $G$ be a graph. If splitter has a winning strategy in the $(\ell,m,d)$-splitter game on $G$, then she has also a winning strategy in the $(\ell\cdot m,1,d)$-splitter game on $G$. If splitter has a winning strategy in the limit $d$-splitter game on $G$, then she has also a strategy in this game where in each round she picks $W_i$ of size $1$.
\end{lemma}

\begin{proof}
Both claims follow from an observation that instead of choosing a set $W$ in some round, splitter can use at most $|W|$ consecutive rounds to choose the vertices of $W$ one by one. The intermediate moves of connector can be ignored because they can only trim further the arena.
\end{proof}

Since for $\ell,m,d\in \Nat$ and a finite graph $G$, the $(\ell,m,d)$-splitter game on $G$ is finite, the fact that splitter does not have a winning strategy in this game implies that connector has one. This statement, however, is not obvious for the limit game. Even though it is not essential for our later reasonings, we state and prove it for the sake of completeness.

\begin{proposition}\label{thm:determinacy}
Let $G$ be a (possibly infinite) graph and let $d\in \Nat$. Then splitter has a winning strategy in the $d$-splitter game on $G$ if and only if connector does not have a winning strategy in this game.
\end{proposition}
\begin{proof}
The $d$-splitter game on $G$ can be regarded as a Gale-Stewart game on $A:=P_{\text{fin}}(V(G))$, the set of finite subsets of $V(G)$, with an open winning condition. Namely, the players alternately choose consecutive entries of a sequence from $A^\omega$. Some moves are invalid for the players: for connector, it is invalid to choose a subset of cardinality different than $1$, and for both players it is invalid to choose any vertex outside the current arena. We assume that if some player performs an invalid move, then the game continues indefinitely with this player being the losing one at the end. Let $X_C$ and $X_S=A^\omega\setminus X_C$ be the sets of all sequences from $A^\omega$ that encode the plays winning for connector and splitter, respectively. Then $X_S$ is the set of all the infinite sequences where the first invalid move was performed by connector. It follows that $X_S$ is open in the product topology on $A^\omega$. Hence, the claim follows from Martin's theorem~\cite{Martin75} that states that Gale-Stewart games with a Borel winning condition are determined.
\end{proof}

\subsection{Linking the standard and limit games}

We first prove a simple auxiliary result.

\begin{lemma}\label{lem:fo-winning}
For every $\ell,m,d\in \Nat$, there exists a first-order sentence $\phi_{\ell,m,d}$ such that, for any graph $G$, $G\models \phi_{\ell,m,d}$ iff connector wins in the $(\ell,m,d)$-splitter game on $G$.
\end{lemma}
\begin{proof}
We simply quantify the existence of the strategy for connector. That is, $\phi_{\ell,m,d}$ has $2\ell$ alternating blocks of quantification, where connector's moves are quantified existentially as single vertices $v_i$, and splitter's moves are quantified universally as $m$ variables $w^1_i,w^2_i,\ldots,w^m_i$ (not required to be distinct). To verify that the quantified variables form a valid play of the $(\ell,m,d)$-splitter game on $G$, and at the end the arena is non-empty, we use formulas $\psi_i$ for $i=0,1,\ldots,\ell$ with the following syntax and semantics: Formula $\psi_i$ has free variables $ \{v_j\}_{j\leq i}\cup \{w^q_j\}_{q\leq m,\, j\leq i}\cup \{x\}$, and is true if and only if $x$ belongs to arena $G_i$ defined by the previous moves $\{v_j\}_{j\leq i}\cup \{w^q_j\}_{q\leq m,\, j\leq i}$. We can put $\psi_0=\top$, and $\psi_{i+1}$ can be defined inductively using $\psi_i$ as follows: $x$ belongs to $G_{i}$, $x$ is different from all $\{w^q_{i+1}\}_{q\leq m}$, and moreover there is a path in $G_i$ of length at most $d$ from $x$ to $v_{i+1}$.
\end{proof}

\begin{lemma}\label{lem:game}
Let $\cl$ be a class of finite graphs. Then $\cl$ is winning for splitter if and only if $\cl$ is limit winning for splitter.
\end{lemma}
\begin{proof}
Suppose first that $\cl$ is winning for splitter. Fix any $d\in \Nat$. Then there exist $\ell,m\in \Nat$ such that splitter wins in the $(\ell,m,d)$-splitter game on each graph of $\cl$. We claim that splitter wins also in the $(\ell,m,d)$-splitter game on each graph of $\cls$, so in particular also in the limit $d$-splitter game on each graph of $\cls$. Indeed, take any graph $H\in \cls$; then $H$ is a subgraph of a graph of the form $G^\star = \ultraprod G_n$ for some sequence $(G_n)_{n\in \omega}$ of graphs from $\cl$. To show that splitter wins in the $(\ell,m,d)$-splitter game on $H$, it suffices to show that she wins in this game on $G^\star$, because then the same strategy trimmed to $V(H)$ will work on $H$. However, this follows immediately from Lemma~\ref{lem:fo-winning} and Łoś's theorem: we have that $G_n\models \neg \phi_{\ell,m,d}$ for all $n\in \Nat$, so by Łoś's theorem also $G^\star\models \neg \phi_{\ell,m,d}$.

Suppose now that $\cl$ is limit winning for splitter. For the sake of contradiction, suppose $\cl$ is not winning for splitter. This implies that there exist $d\in \Nat$ and a sequence of graphs $(G_n)_{n\in \Nat}$ such that, for each $n\in \Nat$, connector has a winning strategy $\sigma_n$ in the $(n,1,d)$-splitter game on $G_n$. Let $G^\star=\ultraprod G_n$. Using $(\sigma_n)_{n\in \Nat}$, we will construct a strategy $\sigma^\star$ for connector in the limit $d$-splitter game on $G^\star$ that wins with all the strategies of splitter, where splitter uses only singleton sets $W_i$. By Lemma~\ref{lem:one-move}, this will be a contradiction with $\cl$ being limit winning for splitter.

Let $(\seq v_1,\seq w_1,\seq v_2,\seq w_2,\ldots,\seq v_i,\seq w_i)$ be the sequence of connector's and splitter's moves during the first $i$ rounds of the limit $d$-splitter game on $G^{\star}$. Here, we assume that splitter always answers with singleton sets, so we denote the $i$-th move of splitter as $\set{\seq w_i}$. Moreover, we suppose that each $\seq v_j$ and $\seq w_j$ is an arbitrarily chosen representative of its equivalence class with respect to $\sim_\ultra$. We denote $\seq v_j=(\seq v_{j,n})_{n\in \Nat}$, where $\seq v_{j,n}\in V(G_n)$ for all $n\in \Nat$, and similarly for $\seq w_j$. 

For $n\in \Nat$ and $i<n$, let 
$$\seq v_{i+1,n}=\sigma_n(\seq v_{1,n},\seq w_{1,n},\seq v_{2,n},\seq w_{2,n},\ldots,\seq v_{i,n},\seq w_{i,n}).$$
For $i\geq n$, define $\seq v_{i+1,n}$ to be an arbitrary vertex of $G_n$. Let us define 
$$\seq v_{i+1}=[(\seq v_{i+1,n})_{n\in \Nat}]_{\sim_\ultra}.$$
Note that the definition of $\seq v_{i+1}$ does not depend on the particular choice of representatives $\seq v_1,\seq w_1,\seq v_2,\seq w_2,\ldots,\seq v_i,\seq w_i$. Put
$$\sigma^\star(\seq v_1,\seq w_1,\seq v_2,\seq w_2,\ldots,\seq v_i,\seq w_i)=\seq v_{i+1}.$$
It remains to prove that $\sigma^\star$ indeed is a strategy that guarantees that connector wins against any strategy of splitter that uses only singleton sets. For this, we need to show that each move $\seq v_{i+1}$ of connector belongs to the arena $G_{i}$, for all $i\in \Nat$. 

To this end, define a first-order formula $\psi_i$ for $i\in \Nat$ as in the proof of Lemma~\ref{lem:fo-winning}: $\psi_i$ has free variables $\{v_1,w_1,\ldots,v_i,w_i,x\}$ and is true in $G,v_1,w_1,\ldots,v_i,w_i,x$ if and only if $x$ is in arena $G_i$ defined by previous moves $v_1,w_1,\ldots,v_i,w_i$ of connector and splitter. We can put $\psi_0=\top$, and define $\psi_{i+1}$ using $\psi_i$ inductively as follows: $x$ belongs to arena $G_{i}$, is different from $w_{i+1}$, and there is a path in $G_i$ of length at most $d$ from $x$ to $v_{i+1}$. By the definition of $\sigma_n$, we have that 
$$G_n,\seq v_{1,n},\seq w_{1,n},\seq v_{2,n},\seq w_{2,n},\ldots,\seq v_{i,n},\seq w_{i,n},\seq v_{i+1,n}\models \psi_{i}$$
for all $0\leq i<n$. Therefore, for a fixed $i\in \Nat$ and variable $n\in \Nat$, the above holds \uae. From Łoś's theorem we thus infer that
$$G^\star,\seq v_1,\seq w_1,\seq v_2,\seq w_2,\ldots,\seq v_{i},\seq w_{i},\seq v_{i+1}\models \psi_i$$
for all $i\in \Nat$. Hence, for all $i\in \Nat$ the $(i+1)$-st move of connector is always in the arena $G_{i}$, and thus is valid.
\end{proof}

We remark that by applying the proof of Lemma~\ref{lem:game} in both directions, and then Lemma~\ref{lem:one-move}, we obtain the following corollary: If splitter wins in the limit $d$-splitter game on every graph from $\cls$, then there is a constant $\ell\in \Nat$ such that she wins in the $(\ell,1,d)$-splitter game on every graph from $\cls$. In other words, the lengths of the game can be bounded universally.

\subsection{Equivalence with nowhere-denseness and quasi-wideness}

The proofs in this section are inspired by the analogous proofs for the finite variants, due to Grohe et al.~\cite{GroheKS14}.

\begin{lemma}\label{lem:qw-gm}
If a class of finite graphs $\cl$ is uniformly limit quasi-wide, then $\cl$ is limit winning for splitter.
\end{lemma}
\begin{proof}
Fix a graph $G\in \cls$ and $d\in \Nat$. We give a winning strategy $\tau$ for splitter in the limit $d$-splitter game on $G$. Suppose $(v_1,W_1,v_2,W_2,\ldots,v_{i-1},W_{i-1},v_i)$ were the previously performed moves in the game. Since connector's move $v_i$ is within the arena $G_{i-1}$ and the arenas in the game form a chain in the subgraph order, for each $1\leq j<i$ there exists a path $P_{j,i}$ in arena $G_{j-1}$ that has endpoints $v_j$ and $v_i$, and length at most $d$. Then splitter answers with move $W_i:=(\{v_i\}\cup \bigcup_{0\leq j<i} V(P_{j,i}))\cap V(G_{i-1})$; note that thus $W_i$ is finite and contains $v_i$.

Suppose that this strategy is not winning for splitter, that is, there exists an infinite sequence $(v_1,v_2,\ldots)$ of connector's moves that are always within the arena. Since $v_i\in W_i$ for each $i=1,2,\ldots$, all $v_i$-s are pairwise different. Let $M=\{v_1,v_2,\ldots\}$; then $M$ is infinite. As $G\in \cls$ and $\cl$ is uniformly limit quasi-wide, there is a finite set $S\subseteq V(G)$ and infinite subset $M'\subseteq M\setminus S$ that is $d$-scattered in $G\setminus S$.

For $n\in \Nat$, let $R_n=\bigcup_{1\leq j<i\leq n} V(P_{j,i})$. Since $(R_n)_{n\in \Nat}$ is an increasing inclusion chain and $S$ is finite, there exists $N\in \Nat$ such that $S\cap (R_m\setminus R_N)=\emptyset$ for all $m\geq N$. Note that by the definition of the sets $W_i$ it follows that $R_N\cap V(G_N)=\emptyset$. Since $M'$ is infinite, there exist two distinct indices $m>n>N$ such that $v_{m},v_{n}\in M'$. Consider the path $P_{n,m}$. This path is entirely contained in arena $G_{m-1}\subseteq G_N$, and hence $R_N\cap V(P_{n,m})=\emptyset$. Therefore, $V(P_{n,m})\subseteq R_{m}\setminus R_N$, and hence $P_{n,m}$ does not traverse any vertex of $S$. However, $P_{n,m}$ has length at most $d$ and connects two different vertices of $M'$. This is a contradiction with $M'$ being $d$-scattered in $G\setminus S$.
\end{proof}

\begin{lemma}\label{lem:gm-tnd}
If a class of finite graphs $\cl$ is limit winning for splitter, then $\cl$ is limit topologically nowhere dense.
\end{lemma}
\begin{proof}
The proof is by the contraposition. Suppose there is $d\in\Nat$ such that $K_\omega\in \cls\topnab d$. Since $\cls$ is closed under taking subgraphs by definition, we have that there is a graph $H\in \cls$ such that $H$ can be obtained from $K_\omega$ by subdividing each edge at most $2d$ times. We now present a strategy for connector in the limit $(4d+2)$-splitter game.

Let $A$ be the set of vertices of infinite degree in $H$. Every pair of vertices $\{u,v\}\subseteq A$ can be connected by an infinite family of internally vertex-disjoint paths of length at most $4d+2$. Hence, for every finite set $W\subseteq V(G)$, all the vertices of $A\setminus W$ belong to the same connected component of $H\setminus W$, and moreover they are pairwise at distance at most $4d+2$ in $H\setminus W$. 

The strategy for connector is as follows: in each round, pick an arbitrary vertex of $A$ that was not chosen by splitter in the previous rounds. Such a vertex always exists, because $A$ is infinite and after any finite number of rounds, only a finite number of vertices have been chosen by splitter. If connector follows this strategy, then the argument of the previous paragraph shows that all the vertices of $A$ apart from the ones chosen by splitter will be always in the arena. Hence the moves of connector will be always valid and the arena will never become empty.
\end{proof}

\section{Conclusions}\label{sec:conc}
All the implications presented in the previous sections, depicted on Figure~\ref{fig:diagram}, together prove the following result.

\begin{theorem}\label{thm:main}
Let $\cl$ be a class of finite graphs. Then the following conditions are equivalent.
\begin{itemize}
\item[(i)] $\cl$ is nowhere dense;
\item[(ii)] $\cl$ is topologically nowhere dense;
\item[(iii)] $\cl$ is limit nowhere dense;
\item[(iv)] $\cl$ is limit topologically nowhere dense;
\item[(v)] $\cl$ is strongly limit quasi-wide with some margin $s$;
\item[(vi)] $\cl$ is strongly uniformly limit quasi-wide with some margin $s$;
\item[(vii)] $\cl$ is limit quasi-wide;
\item[(viii)] $\cl$ is uniformly limit quasi-wide;
\item[(ix)] $\cl$ is winning for splitter.
\item[(x)] $\cl$ is limit winning for splitter.
\end{itemize}
Moreover, if $\cl$ is hereditary, then the following conditions are also equivalent with the ones above:
\begin{itemize}
\item[(xi)] $\cl$ is quasi-wide with some margin $s$;
\item[(xii)] $\cl$ is uniformly quasi-wide with some margin $s$;
\end{itemize}
\end{theorem}

Thus, the proof of Theorem~\ref{thm:main} gives also an alternative proof of the existing results of the theory of sparse graphs, namely the equivalence of 
\begin{itemize}
\item[(a)] (topological) nowhere denseness of $\cl$;
\item[(b)] (uniform) quasi-wideness of $\cl$;
\item[(c)] $\cl$ being winning for splitter.
\end{itemize}
The main advantage of the presented technique is that it simplifies the conceptually difficult part of the reasoning. One could argue that the direct proofs in the standard setting are technical translations of more natural reasonings about properties of infinite graphs. However, this translation using Łoś's theorem, while usually being conceptually straightforward, requires some technical effort and can lead to caveats, as is the case for quasi-wideness.

Another drawback of the presented methodology is that it is inherently non-constructive and can give only qualitative results. In a direct proof in the standard setting, say of the equivalence between nowhere denseness and quasi-wideness, one can usually trace the relation between the parameters of the notions involved. Proofs that use Łoś's theorem and a reasoning in the limit setting can only prove the existence of some finite parameters, and cannot be leveraged to show any bounds on their magnitudes.

As for the future work, there are multiple other, equivalent definitions of nowhere denseness that use edge density in shallow minors of $\cl$, or various parameters connected to coloring and ordering graphs from $\cl$~\cite{sparsity}. In these cases, it is most common that nowhere denseness of $\cl$ is equivalent to a bound of the form $\chi(G)\leq |V(G)|^\varepsilon$, where $\chi(\cdot)$ is a relevant graph parameter, holding for every $\eps>0$ and a sufficiently large graph $G\in \cl$. Definitions of these form can be in principle translated to the ultraproduct setting using non-standard arithmetics, but we failed to observe any conceptual gain that would emerge from such considerations. Therefore, it remains open to find a meaningful extension of the presented methodology to these quantitative definitions of nowhere denseness.

\bibliographystyle{abbrv}
\bibliography{main}

\end{document}